\documentclass[12pt]{article}

\usepackage{graphicx}
\usepackage{subcaption}
\usepackage{xcolor}

\usepackage{amsmath,amssymb,epsf,epic}

\newcommand{\ra}{\rightarrow}

\newcommand{\bra}{\langle} 

\newcommand{\ket}{\rangle}

\newcommand{\E}{{\mathbb E}}
\newcommand{\D}{{\mathbb D}}

\newcommand{\be}{\begin{equation}}
\newcommand{\ee}{\end{equation}}
\newcommand{\bea}{\begin{eqnarray}}
\newcommand{\eea}{\end{eqnarray}}

\newcommand{\ep}{\hfill  {\vrule height 10pt width 8pt depth 0pt}}

\newcommand{\grintl}{[\kern-.18em [}
\newcommand{\grintr}{]\kern-.18em ]}

\newcounter{resultcounter}[section]

\newtheorem{thm}[resultcounter]{Theorem}
\newtheorem{lem}[resultcounter]{Lemma}
\newtheorem{prop}[resultcounter]{Proposition}

\newtheorem{definition}[resultcounter]{Definition}
\newtheorem{rem}[resultcounter]{Remark}
\newtheorem{rems}[resultcounter]{Remarks}

  \def\cC{{\cal C}}

 \def\cK{{\cal K}} \def\cL{{\cal L}}

 \def\cT{{\cal T}}

\newcommand{\R}{{\mathbb R}}
\newcommand{\N}{{\mathbb N}}

\newcommand{\C}{{\mathbb C}}
\newcommand{\G}{{\mathbb G}}
\newcommand{\Z}{{\mathbb Z}}
\renewcommand{\E}{{\mathbb E}}
\renewcommand{\P}{{\mathbb P}}
\newcommand{\I}{{\mathbb I}}
\newcommand{\T}{{\mathbb T}}

\begin{document}
\title{Lower Bounds on the Localisation Length of \\ Balanced Random Quantum Walks}
 
\author{Joachim Asch
\thanks{
Aix Marseille Univ, Universit\'e de Toulon, CNRS, CPT, Marseille, France} ,
Alain Joye
\thanks{
Univ. Grenoble Alpes, CNRS, Institut Fourier, F-38000 Grenoble, France}
}

\date{April 5, 2019 }

\maketitle
\vspace{-1cm}

\thispagestyle{empty}
\setcounter{page}{1}
\setcounter{section}{1}

\setcounter{section}{0}

\abstract{ We consider the dynamical properties of  Quantum Walks defined on the $d$-dimensional cubic lattice, or the homogeneous tree of coordination number $2d$, with site dependent random phases, further characterised by transition probabilities between neighbouring sites equal to $1/(2d)$. We show that the localisation length for these Balanced Random Quantum Walks can be expressed as a combinatorial expression involving sums over weighted paths on the considered graph. This expression provides lower bounds on the localisation length by restriction to paths with weight $1$, which allows us to prove the localisation length diverges on the tree as $d^2$. On the cubic lattice, the method yields the lower bound $1/\ln(2)$ for all $d$, and allows us to bound the localisation length from below by the correlation length of self-avoiding walks computed at $1/(2d)$.}

\thispagestyle{empty}
\setcounter{page}{1}
\setcounter{section}{1}

\setcounter{section}{0}

\section{Introduction}

The last decade has seen Quantum Walks, (QWs for short), deterministic or random, play a growing role in several scientific fields, starting with the modeling of quantum dynamics in various physical situations. Let us simply mention that QWs provide a simple description of the electronic motion in the Quantum Hall geometry \cite{cc, KOK, ABJ2}, while their effectiveness in quantum optics as discrete models to study atoms trapped in time periodic optical lattices or polarized photons propagating in networks of waveguides has found experimental confirmation \cite{Zetal, scia}. Theoretical quantum computing is another field in which QWs prove to be useful as tools in the elaboration or assessment of quantum algorithms, see \cite{sa, MNRS, P}. From a mathematical perspective, QWs have been considered as non commutative analogs of classical random walks, see {\it e.g.} \cite{Ko, skj, gnvw}, or as subject of study in the spectral analysis of (random) unitary operators, in the line of \cite{bhj, cmv, J1, HJS, dfv}. See the reviews and book \cite{va, J3, S1, ABJ3} for more applications, mathematical results, and references. 

The present paper addresses the mathematical (de)localisation properties of certain Random Quantum Walks (RQWs for short) defined on infinite graphs, for the discrete time random quantum dynamical system defined by the iterates of the RQW unitary operator. Anderson localisation is ubiquitous in the physics of disordered systems and has been the object of intense research since its discovery in the fifties \cite{ACAT}, see {\it e.g.}  \cite{cl, Ki, AW2}. Delocalisation of quantum particles in presence of disorder, although expected to be present in high dimensional lattices, is however difficult to exhibit; it has been proven to hold for the Anderson model on the tree, or mild variants thereof \cite{Kl, W}. In this context, certain RQWs viewed as substitutes of the random evolution operators generated by the discrete Anderson Hamiltonians, have proven to be mathematically more tractable and to display analogous behaviours regarding the transport properties: these RQWs describe the motion of a particle with internal degree of freedom, or quantum walker, hopping on the sites an infinite underlying graph $\G^d$, standing for the cubic lattice $\Z^d$ or $\cT_{2d}$, the homogeneous tree of coordination number $2d$, in a static random environment. The internal degree of freedom, or coin state, lives in $\C^{2d}$ and the deterministic part of the walk is defined as follows. Given a natural basis in $\C^{2d}$, the one time step unitary evolution $U(C)$ is obtained by the action of a unitary matrix $C\in U(2d)$ on the coin state of the particle, followed by the action of a coin state conditioned shift $S$ which moves the particle to its nearest neighbours on $\G^d$.  Static disorder is introduced in the model via {\em  i.i.d. random phases} used to decorate the coin matrix $C$ in such a way that the unitary coin state update becomes {\em site-dependent} and random. The coin matrix $C$ of the resulting random unitary operator $U_\omega(C)$, called the {\it skeleton}, is a parameter which, in a sense, monitors the strength of the disorder; see the next section for precise definition. 

Restricting attention to RQWs constructed this way,  we recall some of their properties. RQWs $U_\omega(C)$ defined on the one dimensional lattice exhibit dynamical localisation for all non-trivial choice of skeleton $C$, \cite{JM, ASW},  see also \cite{Betal} for results on CMV matrices by a different approach. Moreover, RQWs that have skeleton close to certain permutation matrices are known to exhibit dynamical localisation on the cubic lattice $\Z^d$, $d>1$  \cite{J2,J3}; a result analogous to the strong disorder regime in the Anderson model framework. To complete the analogy, RQW's $U_\omega(C)$ defined on homogeneous trees are proven to undergo a localisation-delocalisation spectral transition in \cite{HJ1}, when the skeleton $C$ moves from neighbourhoods of certain permutation matrices  to neighbourhoods of other permutation matrices, akin to the spectral transition occurring for the Anderson model on the tree as a function of the strength of the disorder \cite{W}. See also \cite{ABJ4} for delocalisation results of a topological nature for the Chalker Coddington model.

The goal of the present paper is to analyse the unitary operator $U_\omega(C)$ in case the skeleton $C$ lies as far as possible from the permutation matrices, for which  localisation results hold, and to investigate the possible delocalisation properties of those RQWs. Therefore we consider RQWs associated to {\it balanced} skeletons $C\in U(2d)$ characterised by the fact that in the natural basis the moduli of its entries are all equal to $1/\sqrt {2d}$. This amounts to saying that the quantum mechanical transition probabilities induced by $U_\omega(C)$ between neighbouring sites are all equal to $1/(2d)$, irrespective of the realisation of the random phases decorating the balanced skeleton. Such RQWs are called {\it Balanced Random Quantum Walks}, BRQWs for short. 
That BRQWs are likely to exhibit delocalisation properties on the cubic lattice is not granted. However, numerics obtained for the Chalker-Coddington model suggest delocalisation for values of the parameters of the model that correspond, essentially, to the balanced behaviour introduced above, see \cite{KOK}; hence, the analogy is tempting. 
When considered on the tree $\cT_{2d}$, BRQWs are thus neither related to localisation, nor to delocalisation, {\it a priori}, and correspond to values of the parameter $C$ that are not covered by the analysis of \cite{HJ1}. 

Our results concern {\it lower bounds on the localisation length} of BRQWs $U_\omega(C)$ defined on the cubic lattice $\Z^d$ or the tree $\cT_{2d}$, making use of two characteristic specificities of BRQWs operators: uniform transition probabilities between neighbouring sites and uniform distribution of the iid random phases. The localisation length $\cL$ we consider is defined, informally, as the inverse of the largest $\alpha>0$ such that 
$\limsup_{n\ra\infty}\E(\|e^{\alpha |X|/2}U_\omega^n(C) \ \psi_0\|^2)<\infty$, for all initial states $\psi_0$ supported on a single site of $\G^d$ with given internal degree of freedom, 
and $|X|$ is the multiplication operator by the norm of the position on $\G^d$. The localisation length defined similarly for compactly supported initial states is then bounded below by the one we consider. The fact that RQWs couple nearest neighbours on $\G^d$ only allows for a head on approach of the localisation length.

BRQWs are essentially parameter free models, in the sense that once the entries of the skeleton have fixed amplitudes $1/\sqrt{2d}$, only their (correlated) phases remain. Hence, we investigate the large $d$ limit of $\cL$, that we will dub (improperly for the tree case) the large dimension limit, a regime in  which delocalisation is intuitively more likely to occur, so that $\cL$ is more likely to be large. Our first contribution to this question is a reformulation of the problem in terms of a combinatorics problem on the underlying graph $\G^d$, presented in Section \ref{btoc} . Within this framework, lower bounds on the localisation length are obtained via the analysis of partition functions defined on certain subsets of paths on $\G^d$, with weights related to the phases of the entries of the balanced skeleton in $U_\omega(C)$. Using this framework for the homogeneous tree $\cT_{2d}$ in Section \ref{improvedtree}, we get the large $d$ behaviour $\cL \geq 2d^2(1+ O(1/d))$ in Theorem \ref{locletree}. We discuss this approach in  the case of the cubic lattice $\Z^d$ in Section \ref{lbszd}. We cannot show that the localisation length diverges with the dimension $d$, but get the lower bound $\cL\geq 1/\ln(2)$, for all $d\geq 1$, see Theorem \ref{thmone}. From a more general perspective, we relate the partition functions used to bound the localisation length to partition functions or susceptibilities considered in the study of self-avoiding walks in $\Z^d$ in the closing Section \ref{seccor}. In particular, we show in Proposition \ref{corloc} that that the localisation length satisfies $\cL\geq \xi_d(1/(2d))$, where $\xi_d(1/(2d))$ is the correlation length of self-avoiding walks in $\Z^d$, at the critical value of simple random walks. The large dimension analysis of this expression remains to be performed which, according to experts in the field, represents a nontrivial task.\\

\noindent
{\bf\large  Acknowledgments} \vspace{.2cm}

This work is partially supported by the ANR grant NONSTOPS (ANR-17-CE40-0006-01), and by the program ECOS-Conicyt C15E10. We acknowledge the warm hospitality and support of the Centre de Recherches Math\'ematiques in Montr\'eal during the last stage of this work: A.J. was a CRM-Simons Scholar-in-Residence Program, and J.A. was a guest of the Thematic Semester Mathematical Challenges in Many Body Physics and Quantum Information. We wish to thank V. Beffara, L. Coquille, and S. Warzel for several enlightening discussions and G. Slade for correspondance, at various stages of this work.  Also, we are grateful to the referees for constructive remarks.

\section{Balanced Random Quantum Walk}

By Balanced Random Quantum Walk, BRQW for short, we mean a quantum walk defined either on the cubic lattice $\Z^d$ or on the homogeneous tree $\cT_{2d}$, where $d\in \N$ is half the coordination number in the latter case, such that the quantum mechanical transition amplitude between neighbouring sites has uniform modulus and random argument uniformly distributed on the torus.  We briefly recall the basics about RQWs below; for more details, see \cite{HJ1, HJ2} \bigskip

\subsection{Random quantum walks on $\cT_{2d}$}
Let $\cT_{2d}$ denote the homogeneous tree of degree $2d$ of the free group $F_{\{a_1, \dots, a_d\}}$ generated by the alphabet
\be
A_{2d}=\{a_1, \dots, a_d,  a_1^{-1}, \dots, a_d^{-1}\},
\ee 
with $a_j a_j^{-1}=a_j^{-1}a_j=e$, $j=1,\dots, d$, $e$ being the neutral element of the group. 
A vertex of $\cT_{2d}$ denoted by $e$ is chosen to be the root of the tree. Each vertex $x=x_{1}x_{2}\dots x_{n}$, $n\in\N$ of $\cT_{2d}$ is a reduced word of finitely many letters from the alphabet $A_{2d}$ and an edge of $\cT_{2d}$ is a pair of vertices $(x,y)$ such that $xy^{-1}\in A_{2d}$. Any pair of vertices $x$ and $y$ can be joined by a unique set of edges, or path in $\cT_{2d}$ and the number of nearest neighbours of any vertex in $\cT_{2d}$ is thus $2d$. 
We identify $\cT_{2d}$ with its set of vertices, and define the configuration Hilbert space of the walker by
\be
l^2(\cT_{2d})=\Big\{\psi=\sum_{x\in \cT_{2d}}\psi_x |x\ket \ \mbox{s.t.} \  \psi_x\in \C,  \ \sum_{x\in \cT_{2d}}|\psi_x|^2<\infty\Big\},
\ee
where $|x\ket$ denotes the element of the canonical basis of $l^2(\cT_{2d})$ which sits at vertex $x$.
The coin Hilbert space (or spin Hilbert space) of the quantum walker on $\cT_{2d}$ is $\C^{2d}$. The elements of the ordered canonical basis of $\C^{2d}$ are labelled by letters in the alphabet $A_{2d}=\{a_1, \dots, a_d,  a_1^{-1}, \dots, a_d^{-1}\}$, and
the total Hilbert space is
\be\label{canbas}
\cK_{d} = l^2(\cT_{2d})\otimes \C^{2d} \ \mbox{ with canonical basis }\  \big\{x\otimes \tau\equiv |x\ket\otimes |\tau\ket ,\ \ x\in \cT_{2d}, \tau\in A_{2d}\big\}.
\ee
The quantum walk on the tree is defined as the composition of a unitary update of the coin (or spin) variables in $\C^{2d}$ followed by a coin  state dependent shift on the tree. Let $C$ be a unitary matrix on $\C^{2d}$, {\it i.e.} $C\in U(2d)$. The unitary update operator given by $\I\otimes C$ acts on the canonical basis of $\cK_d$ as
\be\label{reshuffle}
(\I\otimes C) x\otimes \tau=|x\ket\otimes |C\tau\ket=\sum_{\sigma\in A_{2d}} C_{\sigma\tau}\, x\otimes \sigma,
\ee
where $\{C_{\sigma\tau}\}_{(\sigma,\tau)\in A_{2d}^2}$ denote the matrix elements of $C$.
The coin  state dependent shift $S$ on $\cK_{d}$ is defined by
\be\label{dirsumshift}
S={\sum}_{\tau\in A_{2d}} S_\tau\otimes |\tau\ket\bra \tau| 
\ee
where for all $\tau\in A_{2d}$ the unitary operator $S_\tau$ is a shift that acts on $l^2(\cT_{2d})$ as  
\be
S_\tau|x\ket=|x\tau\ket, \ \ \forall \tau \in A_{2d}, \ \ \forall x\in \cT_{2d}.
\ee
Note that $S_\tau^{-1}=S_\tau^*=S_{\tau^{-1}}$. 
A homogeneous quantum walk on $\cT_{2d}$ is then defined as the one step unitary evolution operator on $\cK_d=l^2(\cT_{2d})\otimes \C^{2d}$ given by 
\be\label{defeven}
U(C)=S(\I\otimes C)
={\sum_{\tau\in A_{2d}, x\in \cT_{2d}} |x\tau\ket\bra x |\otimes |\tau\ket\bra \tau| C},
\ee
where $C\in U({2d})$ is a parameter.

Consider now a family of coin matrices $\cC=\{C(x)\in U({2d})\}_{x\in\cT_{2d}}$, indexed by the vertices $x\in\cT_{2d}$. We generalise the construction to quantum walks with site dependent coin matrices by means of the definition
\be\label{explicitsitedep}
U(\cC)=\sum_{\tau\in A_{2d},
 x\in \cT_{2d}} |x\tau\ket\bra x |\otimes |\tau\ket\bra \tau|  C(x).  
\ee
In order to deal with RQWs, we introduce a probability space $\Omega=\T^{\cT_{2d}\times A_{2d}}$,  $\T=\R/(2\pi \Z)$, with $\sigma$ algebra generated by the cylinder sets and measure $\P=\otimes_{x\in\cT_{2d}\atop \tau\in A_{2d}}d\nu$ where $d\nu$ is a probability measure on $\T$. 
Let $\{\omega_{x,\tau}\}_{x\in \cT_{2d}, \tau\in A_{2d} }$ be a set of i.i.d. random variables on the torus $\T$ with common distribution $d\nu$. We will note $\Omega\ni \omega=\{\omega_{x, \tau}\}_{x\in \cT_{2d}, \tau\in A_{2d} }$.\bigskip

The RQWs we consider are constructed by means of families of site dependent random coin matrices. Let  $\cC_\omega=\{C_\omega(x)\in U({2d})\}_{x\in \cT_{2d}}$ be the family of random coin matrices depending on a fixed matrix $C\in U({2d})$, the {\it skeleton},  where, for each $x\in \cT_{2d}$, $C_\omega(x)$ is defined by its matrix elements 
\be
C_\omega(x)_{\tau \sigma}=e^{i\omega_{x\tau, \tau}}C_{\tau\sigma}, \ \ (\tau,\sigma)\in A_{2d}\times A_{2d}.
\ee 
The site dependence appears only in the random phases of the matrices $C_\omega(x)$, which have a fixed skeleton $C\in U({2d})$. The RQWs considered  depend parametrically on $C\in U({2d})$ and are defined by the random unitary operator
\be\label{defrqw}
U_\omega(C):=U(\cC_\omega) \ \mbox{on} \ \cK=l^2(\cT_{2d})\otimes \C^{2d}.
\ee
Observe that defining a random diagonal unitary operator on $\cK_d$ by 
\be\label{dd}
\D_\omega x\otimes \tau = e^{i\omega_{x, \tau}} x\otimes \tau, \ \ \forall (x,\tau)\in \cT_{2d}\times A_{2d},
\ee
we have the identity
\be\label{iddef}
U_\omega(C)=\D_\omega U(C) \ \ \mbox{on } \cK_d,
\ee
which confirms that $U_\omega(C)$ is manifestly unitary. Note also that the transition amplitudes induced by $U_\omega(C)$ are nonzero for nearest 
neighbours on the tree only and 
\be\label{tramp}
\bra x\otimes \tau | U_\omega(C) y\otimes \sigma\ket=e^{i\omega_{x,\tau}}\bra x\otimes \tau | U(C) y\otimes \sigma\ket.
\ee
\begin{rem}
RQWs of a similar flavour  defined on trees with odd coordination number are studied  in \cite{HJ1}.
\end{rem}

\subsection{Random quantum walks on $\Z^d$}

The definition of a RQW $U_\omega(C)$ on $\Z^d$ instead of $\cT_{2d}$ is quite similar: the sites $x\in\cT_{2d}$ are replaced by $x\in\Z^d$ so that the configuration space $l^2(\cT_{2d})$ is replaced by $l^2(\Z^d)$ but the coin space remains the same, {\it i.e.} $\C^{2d}$. Hence the complete Hilbert space is $\cK_d=l^2(\Z^d)\otimes\C^{2d}$ and the update operator $\I\otimes C$ is the same as on the tree. The definition of the shifts $S_\tau$ in $S=\sum_{\tau\in A_{2d}}S_\tau \otimes |\tau\ket\bra\tau|$, see (\ref{dirsumshift}), needs to be slightly adapted.
We associate the letters $\tau$ of the alphabet $A_{2d}$ with multiples of the canonical basis vectors $\{e_1, \dots, e_d\}$ of $\R^d$ as follows
\be
a_\tau  \leftrightarrow e_\tau,\ a^{-1}_{\tau}  \leftrightarrow -e_\tau, \ \ \tau\in\{1,\dots, d\},
\ee
and define the action of $S_\tau$ on $l^2(\Z^d)$ accordingly: for any $x=(x_1, \dots ,x_d)\in \Z^d$
\be
S_{a_\tau}|x\ket=|x+e_\tau\ket, \ S_{a_\tau^{-1}}|x\ket=|x-e_\tau\ket, \ \ \tau\in\{1,\dots, d\}.
\ee
We shall sometimes abuse notations and write for short $S_\tau |x\ket=|x+\tau\ket$, $\tau\in A_{2d}$. The random quantum walk is then defined by $U_\omega(C)$, as in (\ref{iddef}).

\bigskip

The use of the same symbol $\cK_d$ for the total Hilbert space is no coincidence: we deal with the cases of the tree and cubic lattice in parallel in what follows, denoting by $\G^d$ the graph corresponding either to $\Z^d$ or $\cT_{2d}$.  We shall use the notations introduced for the case of the tree only, with the understanding that the replacements just described yield the corresponding statements for the lattice case.

\subsection{Localisation Length}

We will say that {\it averaged dynamical localisation} occurs when there exists an $\alpha>0$ such that for  $\psi_0=e\otimes \tau_0\in \cK_d$ 
\be\label{exploc}
\sup_{n\in \Z}\E(\|e^{\alpha |X|/2 }U_\omega^n(C)\psi_0\|^2) < \infty, 
\ee
where $X$ denotes the position operator on $\G^d$, and $|X|$ is the multiplication operator given by either the distance to the root of the point of $\cT_{2d}$ considered, or the $l^\infty$ or $l^1$ norm of the point considered in $\Z^d$. Averaged localisation may happen on certain spectral subspaces of $U_\omega(C)$ only, in which case $\psi_0$ is restricted to these subspaces in (\ref{exploc}). 
This notion is weaker than exponential dynamical localisation in which, essentially, the $\sup_{n\in\N}$ is inside the expectation, see Remark \ref{adldl} below. It turns out to be useful in defining the quantitative criterion we need in our investigation of the delocalisation properties of RQWs.  

We define the (dynamical) localisation length of  a RQW on $\G^d$ as follows.
\begin{definition} The localisation length $\cL$ of the model is given by $\cL=1/\alpha_s$, where
\be
\alpha_s=\sup \{\alpha\geq 0 \ | \ \mbox{s.t.} \ (\ref{exploc}) \ \mbox{holds}\}
\ee
and with the convention $\infty=1/0$.
\end{definition}

Considering $\psi_0=e\otimes \tau_0$, we will thus analyse the large $n$ behaviour of the expectation of 
\be\label{laplace}
\|e^{\alpha |X|/2}U_\omega^n(C) \ e\otimes\tau_0\|^2
=\sum_{x, \tau}e^{|x|\alpha}|\bra x\otimes \tau |U_\omega^n(C)\ e\otimes\tau_0\ket|^2,
\ee
where $\sum_{\tau\in A_{2d}}|\bra x\otimes \tau |U_\omega^n\ e\otimes\tau_0\ket|^2$ denotes the quantum mechanical probability to find the quantum walker on site $x\in\G^d$, knowing it started at time zero on the root $e$, with coin state $\tau_0$. Hence,
if the expectation of the right hand side of equation (\ref{laplace}) is finite uniformly in $n\in \N$, for some $\alpha>0$, it means that the limiting averaged distribution of the position of the quantum walker has moments of exponential order and that the localisation length $\cL$ satisfies $\cL<1/\alpha$. 
\bigskip

Our aim is to exhibit an $\alpha_c\geq 0$, as small as possible, such that $\alpha\geq \alpha_c$ implies 
\be\label{alpha_c}
\sup_{n\in \Z}\E(\|e^{\alpha |X|/2}U_\omega^n(C) \ e\otimes\tau_0\|^2)=\infty,
\ee 
for certain choices of coin matrix $C$ and phase distributions, that we will call balanced.
Hence, if averaged dynamical localisation takes place for some $\alpha$, then 
$\alpha<\alpha_c$, so that the localisation length of the model is bounded below: $\cL\geq1/\alpha_c$. Moreover, (\ref{alpha_c}) implies that for all 
$\alpha\geq \alpha_c$, 
\be
\E(\sup_{n\in \Z} \|e^{\alpha |X|/2}U_\omega^n(C) \ e\otimes\tau_0\|^2)=\infty,
\ee
which we interpret as a step towards  delocalisation. This step is all the more pertinent that $\alpha_c$ is small, {\it i.e.} $\cL$ is large.
\begin{rem}\label{adldl}
If {\em strong exponential dynamical localisation} holds, characterised by the existence of a dynamical localisation length $0<1/\mu<\infty$ and a constants $c<\infty$ such that, for all $R\geq 1$, 
\be\label{sedl}
\sum_{{x\in \G^d}\atop {|x|\geq R}}\E(\sup_{n\in \Z}|\bra x\otimes \tau | U_\omega^n(C)  \ e\otimes\tau_0\ket |^2)\leq c e^{-\mu R}
\ee 
see \cite{AW2, HJS}, then {\em averaged dynamical localisation} (\ref{exploc}) holds, for all $\alpha < \mu$. Hence, the existence of
$\alpha_c$ such that (\ref{alpha_c}) holds thus also implies the lower bound $1/\mu\ge 1/\alpha_c$.
\end{rem}

\section{From BRQW to Combinatorics}\label{btoc}

We focus on {\em balanced random quantum walks} on $\Z^d$ and $\cT_{2d}$  that we now define.
\begin{definition} A  balanced random quantum walk (BRQW) on $\G^d$ is characterised by a uniform distribution of random phases
\be
d\nu(\theta) =\frac{d\theta}{2\pi},
\ee
and by a balanced skeleton matrix $C^b\in U({2d})$ whose elements satisfy
\be\label{crit}
C^b_{\tau, \tau'}=\frac{e^{i\alpha_{\tau,\tau'}}}{\sqrt{2d}}, \ \, \alpha_{\tau,\tau'}\in \R, \ \ \forall (\tau, \tau')\in A_{2d}^2.
\ee
\end{definition}
\begin{rems}
i) There exist balanced skeleton matrices in all dimensions. An example is provided by the unitary matrix related to the discrete Fourier transform such that $C^b_{j,k}=e^{-i\pi jk/d}/\sqrt{2d}$, where $j,k\in\{0,\dots,2d-1\}$. Another example is provided by Hadamard matrices whose coefficients are equal to $\pm 1/\sqrt{2d}$, which exist for $d=2^k$, $k\in\N$, at least.\\
ii) BRQW are essentially parameter free models, besides the dimension $d$.
\end{rems}
For a BRQW $U_\omega(C^b)$, we have thanks to (\ref{tramp})
\be
\bra x\otimes \tau | U_\omega(C^b) \,y\otimes \sigma\ket=\frac{e^{i(\omega_{x,\tau}+\alpha_{\tau,\sigma})}}{\sqrt{2d}}\delta_{x,y\tau}.
\ee
Thus, the quantum mechanical transition probabilities between neighbouring sites are all equal to $1/(2d)$, and zero otherwise, be it on the tree or the cubic lattice. 
\bigskip
Thus, for BRQWs (noted $U$ below), we have the finite sums, for any $n>1$,
\bea\nonumber
\bra x_n\otimes \tau_n |U^n\ e\otimes\tau_0\ket\hspace{-.3cm}&=&\hspace{-.5cm}\sum_{{x_1,\dots, x_{n-1}}\atop {\tau_1,\dots, \tau_{n-1}}}
\hspace{-.2cm}\bra x_n\otimes \tau_n  | U x_{n-1}\otimes \tau_{n-1}\ket\cdots \bra x_2\otimes \tau_2|U x_1\otimes \tau_1\ket\bra x_1\otimes \tau_1|U e\otimes\tau_0\ket\\
&=&\hspace{-.2cm}\frac{1}{\sqrt{2d}^{n}}\hspace{-.2cm}\sum_{ {\tau_1,\dots, \tau_{n-1}\ \mbox{\tiny s.t.} }\atop {\tau_1\cdots\tau_{n-1}\tau_n=x_n}}
\hspace{-.3cm}e^{i\sum_{j=1}^{n}(\alpha_{\tau_j,\tau_{j-1}}+\omega_{x_j,\tau_{j}})}\ \ \mbox{with} \ \ x_j=\tau_1\cdots\tau_{j-1}\tau_j.
\eea
The constraint on the $\tau_j$s can be expressed by saying that the summation takes place on all paths  of length $n$ on the graph $\G^d$ from the root $e$ to the end point $x_n$, with steps of length one and prescribed last step $\tau_n$. 
\begin{rems}
i) Given the initial point $x_0 = e$ and initial coin state $\tau_0$, it is equivalent to have the sequence of points visited $\{x_1,\dots, x_{n-1}, x_n\}$ or to have the sequence of coin states $ \{\tau_1,\dots, \tau_{n-1}, \tau_n\}$.\\
ii) The random phases $\omega_{x_j, \tau_j}$ are random variables attached to the oriented edge $(x_j,x_{j-1})$, since $x_{j-1}=x_j\tau_j^{-1}$ is defined uniquely by $x_j$ and $\tau_j$. We adopt this point of view from now on, using the notation
\be
\omega_{x_j, \tau_j}=\omega(x_j,x_{j-1}), \ \ \forall \ j=1,\dots, n.
\ee
By contrast, the deterministic phases $\alpha_{\tau_j,\tau_{j-1}}$ depend on the orientation of the pair of edges $(x_{j-1},x_{j-2})$ and $(x_j,x_{j-1})$, (with a fictitious point $x_{-1}=x_0\tau_0^{-1}$, in case $j=1$).
\end{rems}
Consequently, the right hand side of (\ref{laplace}) reads
\be
\frac{1}{{(2d)}^{n}}\sum_{{x_n, y_n}\atop { \color{black} \tau_n, \sigma_n}}\delta_{x_n, y_n}{\color{black}\delta_{\tau_n,\sigma_n}}e^{\alpha |x_n|/2}e^{\alpha |y_n|/2}\hspace{-.3cm}\sum_{ {{\tau_1,\dots, \tau_{n-1}\ \mbox{\tiny s.t.} }\atop {\tau_1\cdots\tau_{n-1}\tau_n=x_n}}
\atop {{\sigma_1,\dots, \sigma_{n-1}\ \mbox{\tiny s.t.} }\atop {\sigma_1\cdots\sigma_{n-1}\sigma_n=y_n}}}
\hspace{-.3cm}e^{i\sum_{j=1}^{n}(\alpha_{\tau_j,\tau_{j-1}}+\omega(x_j,x_{j-1})-\alpha_{\sigma_j,\sigma_{j-1}}-\omega(y_j,y_{j-1}))}, 
\ee
with 
\be\label{cond}
x_j=\tau_1\cdots\tau_{j-1}\tau_j, \ \ y_j=\sigma_1\cdots\sigma_{j-1}\sigma_j, \ \ \forall j\in\{1,\dots, n\}.
\ee
Introducing the notation $\overrightarrow{x_n}=\{e,x_1,\dots, x_{n-1}, x_n\}$ for paths of length $n$, with steps of length one, and keeping
$\tau_j=x^{-1}_{j-1}x_j$ wherever more convenient (and for $\sigma_j=y^{-1}_{j-1}y_j$), we can write  the right hand side of (\ref{laplace}) using 
$\delta_{x_n, y_n}=\sum_{x}\delta_{x, x_n}\delta_{x, y_n}$,  as
\be\label{reform}
\frac{1}{{(2d)}^{n}}\sum_{x, {\color{black} \tau}}e^{\alpha |x|}\hspace{-.0cm}\sum_{ \overrightarrow{x_n}, \overrightarrow{y_n}}
\hspace{-.0cm}\delta_{x, x_n}\delta_{x, y_n}{\color{black} \delta_{\tau, \tau_n}\delta_{\tau, \sigma_n}}e^{i\sum_{j=1}^{n}(\alpha_{\tau_j,\tau_{j-1}}+\omega(x_j,x_{j-1})-\alpha_{\sigma_j,\sigma_{j-1}}-\omega(y_j,y_{j-1}))}, 
\ee
with  (\ref{cond}) again.
At this point, we want take the expectation over the i.i.d. uniformly distributed random phases. We need the following
\begin{definition}\label{defpc} Given a path  in $\G^d$ of length $n$ starting at $e$,  $\overrightarrow{x_n}=\{e,x_1,\dots, x_{n-1}, x_n\}$, we define its {\em phase content }, $PC(\overrightarrow{x_n})$, by
\be
PC(\overrightarrow{x_n})=\{m_{\overrightarrow{x_n}}(y,z), (y,z)\in \G^d\times\G^d\} \ \ \mbox{where} \ m_{\overrightarrow{x_n}}(y,z)=\hspace{-.2cm}\sum_{(x_k,x_{k-1})\subset \overrightarrow{x_n}}\delta_{(x_k,x_{k-1}), (y,z)}
\ee
and we say that two such paths are {\em equivalent} if and only if
\be\label{defequiv}
\overrightarrow{x_n}\sim \overrightarrow{y_n} \Longleftrightarrow PC(\overrightarrow{x_n})=PC(\overrightarrow{y_n}).
\ee
Denoting by $\lbrack \overrightarrow{x_n}\rbrack$ the equivalence class of $\overrightarrow{x_n}$ and by $P^e_n$ the set of all paths of length $n$ starting at $e$ we have
\be
P^e_n=\bigcup_{\lbrack \overrightarrow{x_n}\rbrack \in P^e_n/\sim }\lbrack \overrightarrow{x_n}\rbrack.
\ee
\end{definition}
The integer $m_{\overrightarrow{x_n}}(y,z)$ gives the number of times the oriented edge $(y,z)$ is visited by the path $\overrightarrow{x_n}$. Note  that $m(y,z)=0$ for all but a finite number of edges $(y,z)$, and that $P^e_n$ is finite.

\begin{lem}
With the notations above,
\be\label{phac}
\E(\|e^{\alpha |X|/2}U_\omega^n \ e\otimes\tau_0\|^2)
=\frac{1}{{(2d)}^{n}}\sum_{{x\in \G^d}\atop{\color{black} \tau\in A_{2d}}}e^{\alpha |x|}\hspace{-.3cm}\sum_{\lbrack \overrightarrow{\tilde x_n}\rbrack\in P^e_n/\sim}\left| \sum_{{\overrightarrow{x_n}\in\lbrack \overrightarrow{\tilde x_n}\rbrack}}
\delta_{x, x_n}{\color{black} \delta_{\tau, \tau_n}}e^{i\sum_{j=1}^{n}\alpha_{\tau_j,\tau_{j-1}}}\right|^2.
\ee
\end{lem}
\begin{proof}
By independence and the property $\E(e^{i n \omega(x_j,x_{j-1})})=\delta_{n,0}$ for any $j$, the expectation of (\ref{reform}) is nonzero only if the paths $\overrightarrow{x_n}$ and $\overrightarrow{y_n}$  over which the summation is performed have the same phase contents. Hence
\begin{align}\label{middle}
\E(\|e^{\alpha |X|/2}U_\omega^n &\ e\otimes\tau_0\|^2)=\\ \nonumber
&\frac{1}{{(2d)}^{n}}\sum_{x, {\color{black} \tau}}e^{\alpha |x|}\hspace{-.0cm}\sum_{\overrightarrow{x_n}}
\sum_{{\overrightarrow{y_n}\ \mbox{\tiny s.t.}}\atop {PC(\overrightarrow{y_n})=PC(\overrightarrow{x_n})}}
\hspace{-.5cm}\delta_{x, x_n}\delta_{x, y_n}{\color{black} \delta_{\tau, \tau_n}\delta_{\tau, \sigma_n}}e^{i\sum_{j=1}^{n}(\alpha_{\tau_j,\tau_{j-1}}-\alpha_{\sigma_j,\sigma_{j-1}})}.
\end{align}
Splitting the sum over $\overrightarrow{x_n}\in P^e_n$ into a sum over all distinct equivalence classes $\lbrack \overrightarrow{\tilde x_n}\rbrack$, and using (\ref{defequiv}) in the sum over $\overrightarrow{y_n}$, we get (\ref{phac}).\ep
\end{proof}
\begin{rems} 
0) The computation of the exponential moments of the position operator averaged over the disorder is thus equivalent to a combinatorial problem. \\ 
i) The function $e^{\alpha |x|}$ can be replaced by any other function of the position. The average of the quantum mechanical 
position moments at time $n$ are obtained with $|x|^p$, $p>0$.\\
ii) For a given equivalence class $\lbrack \overrightarrow{\tilde x_n}\rbrack$,  the sum $ \sum_{{\overrightarrow{x_n}\in\lbrack \overrightarrow{\tilde x_n}\rbrack}}
\delta_{x, x_n}{\color{black} \delta_{\tau, \tau_n}}e^{i\sum_{j=1}^{n}\alpha_{\tau_j,\tau_{j-1}}}$ may vanish, as illustrated by the examples below.
\end{rems}
Let us review  some general properties of the combinatorial expression
\be\label{defsn}
S_n^{\tau_0}(\alpha)=\frac{1}{{(2d)}^{n}}\sum_{{x\in \G^d}\atop {\color{black} \tau\in A_{2d}}}e^{\alpha |x|}\hspace{-.3cm}\sum_{\lbrack \overrightarrow{\tilde x_n}\rbrack\in P^e_n/\sim}\left| \sum_{{\overrightarrow{x_n}\in\lbrack \overrightarrow{\tilde x_n}\rbrack}}
\delta_{x, x_n}{\color{black} \delta_{\tau, \tau_n}}e^{i\sum_{j=1}^{n}\alpha_{\tau_j,\tau_{j-1}}}\right|^2\ee
For fixed $n$, $S_n^{\tau_0}(\alpha)$ is finite for any $\alpha\in \C$, and is equal to one if $\alpha=0$. As a function of $\alpha\in\R$, $S_n^{\tau_0}(\cdot)$ is strictly increasing. Moreover, we have the trivial bounds for $\alpha\geq 0$, $1\leq S_n^{\tau_0}(\alpha)\leq e^{n\alpha}$. 
More generally, for a non negative functions $g$ on $\N$ the contribution in (\ref{defsn}) stemming for the summation over $x\in \G^d$ with $|x|\leq g(n)\leq n$ is bounded by
\begin{align}\label{loupboundsn}
\frac{1}{{(2d)}^{n}}\sum_{{{x \in \G^d}\atop {\color{black} \tau\in A_{2d}}} \atop |x|\leq g(n)}e^{\alpha |x|}\hspace{-.3cm}\sum_{\lbrack \overrightarrow{\tilde x_n}\rbrack\in P^e_n/\sim}\left| \sum_{{\overrightarrow{x_n}\in\lbrack \overrightarrow{\tilde x_n}\rbrack}}
\delta_{x, x_n}{\color{black} \delta_{\tau, \tau_n}}e^{i\sum_{j=1}^{n}\alpha_{\tau_j,\tau_{j-1}}}\right|^2
\leq e^{\alpha g(n)}.
\end{align}

Note also that we can consider $\frac1{2d}\sum_{\tau_0\in A_{2d}}S_n^{\tau_0}(\alpha)$, since the existence of a $\alpha_c$ satisfying (\ref{alpha_c}) for this sum implies it satisfies  (\ref{alpha_c}) for some $\tau_0$ as well. Therefore we drop  the superscript $\tau_0$ from the notation.\\

\noindent
{\bf Two-Dimensional Example:}\\
Consider a BRQW on $\G^2$ driven by a $4\times 4$ Hadamard matrix in the basis $\{a_1, a_2, a_1^{-1}, a_2^{-1}\}$
\be\label{exple}
\frac{1}{2}\begin{pmatrix}
-1 & 1 & 1 & 1 \cr
1 & -1 & 1 & 1\cr
1 & 1 & -1 & 1\cr
1 & 1 & 1 & -1
\end{pmatrix}.
\ee
The corresponding quantity $ \sum_{{\overrightarrow{x_n}\in\lbrack \overrightarrow{\tilde x_n}\rbrack}}
\delta_{x, x_n}{\color{black} \delta_{\tau, \tau_n}}e^{i\sum_{j=1}^{n}\alpha_{\tau_j,\tau_{j-1}}}$ thus belongs to $\Z$, and a minus sign is picked up at each step  of the path which doesn't change direction. One readily checks that distinct paths can have the same phase content and are thus equivalent. Moreover, the contribution from a given equivalence class may be zero. 
This is illustrated on Figure 1 by two paths of length six from $e$ to $x$, with prescribed last step $\tau$, belonging to the same equivalence class which is of cardinal 2. The initial coin state $\tau_0$ is not depicted. There is one occurence of two consecutive steps with the  same direction along the left path, whereas there is none along the right path, so that they have opposite contributions. Hence the contributions of this equivalence class vanishes.  
\vspace{.3cm}

\begin{figure}
    \centering
    \begin{subfigure}[b]{0.33\textwidth}
        \includegraphics[width=\textwidth]{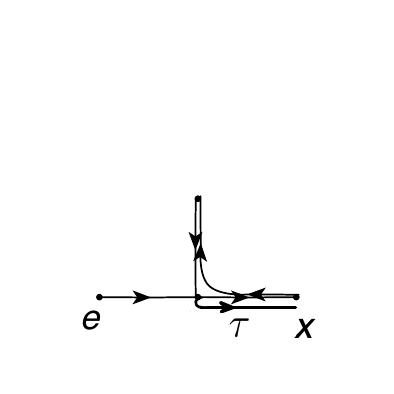}
    \end{subfigure}  \qquad\qquad 
        \begin{subfigure}[b]{0.33\textwidth}
        \includegraphics[width=\textwidth]{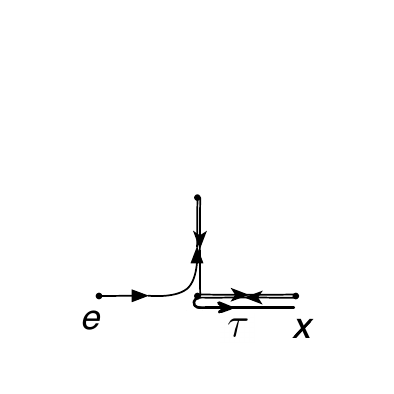}
    \end{subfigure}
    \caption{An equivalence class 
    of paths on  $\mathcal T_{4}$ or $\Z^2$ that does not contribute to (\ref{defsn}).}
\end{figure} 

%\begin{figure}
%    \centering
%    \begin{subfigure}[b]{0.43\textwidth}
%        \includegraphics[scale=0.1]{./treeex.pdf}
%    \end{subfigure}\qquad\hspace{1cm} 
%            \caption{A phase content on $\T_{4}$
%    of two distinct paths only that does not contribute to (\ref{defsn}).}
%\end{figure}

\subsection{Lower Bounds from Restrictions}
In order to avoid dealing with the  combinatorial factor $ \sum_{{\overrightarrow{x_n}\in\lbrack \overrightarrow{\tilde x_n}\rbrack}}
\delta_{x, x_n}{\color{black} \delta_{\tau, \tau_n}}e^{i\sum_{j=1}^{n}\alpha_{\tau_j,\tau_{j-1}}}$, we consider the following  subsets of paths in $P^e_n$:
\begin{definition}
Let $SP_n\subset P^e_n$ be the set of paths the equivalence class of which consists of a {\em single path}, i.e.
\be
\overrightarrow{x} \in SP_n \ \Leftrightarrow \ \# \lbrack \overrightarrow{x}\rbrack =1. 
\ee
Let $SAW_n\subset SP_n$ be the set of self avoiding walks on $\G^d$, i.e. of paths 
$\overrightarrow{x}=\{e,x_1,\dots,x_n\}$ s.t. $x_j\neq x_k$ if $j\neq k$ , and $x_j\neq e$.
\end{definition}
In particular, for $n$ fixed, paths that are sufficiently stretched belong to $SP_n$, whereas self-avoiding walks belong to $SP_n$ by construction. 
%We shall write $\overrightarrow x \sim SP_n$ to mean $\overrightarrow x\in \lbrack \overrightarrow{x_n}\rbrack$, where $x_n\in SP_n$.
Consequently, if $\overrightarrow{\tilde x}_n\in SP_n$,
\be
 \left|\displaystyle\sum_{{\overrightarrow{x_n}\in\lbrack \overrightarrow{\tilde x_n}\rbrack}}
\delta_{x, x_n}{\color{black} \delta_{\tau, \tau_n}}e^{i\sum_{j=1}^{n}\alpha_{\tau_j,\tau_{j-1}}}\right|=\delta_{x, x_n}{\color{black} \delta_{\tau, \tau_n}}.
\ee
Hence, restricting the sum (\ref{defsn}) to equivalence classes of such paths, $S_n(\alpha)$ is bounded from below by quantities that do not depend on the phases of the balanced coin matrix $C^b$:
\be\label{lb}
S_n(\alpha)\geq \frac{1}{{(2d)}^{n}}\sum_{\overrightarrow{x}\in SP_n}e^{\alpha |x_n|}\geq \frac{1}{{(2d)}^{n}}\sum_{\overrightarrow{x}\in SAW_n}e^{\alpha |x_n|}\geq \frac1{2^n}.
\ee
The last inequality corresponds to the very crude lower bound on the number of self-avoiding walks of length $n$ in $\G^d$ given by $d^n$.
We use the notations 
\be\label{partfun}
Z_{SP_n}(\alpha)=\sum_{\overrightarrow{x}\in SP_n}e^{\alpha |x_n|}, \ \mbox{and} \ \ Z_{SAW_n}(\alpha)=\sum_{\overrightarrow{x}\in SAW_n}e^{\alpha |x_n|},
\ee
for those partition functions.
The following properties of these functions are derived in a standard fashion, making use of the fact that they are quite similar to partition functions
of polymer models, see {e.g} \cite{IV}. We provide a proof in Appendix for the convenience of the reader.
\begin{lem}\label{stand}
 Both functions $Z_{X_n}(\alpha)$, $X\in\{SP, SAW\}$, are strictly log-convex functions of $\alpha\in \R^+$ for each $n\geq 1$. For  $\alpha\geq 0$, one has
\be\label{free}
 \Lambda_X(\alpha)=\lim_{n\ra\infty} \frac{\ln Z_{X_n}(\alpha)}{n}=\inf_{n\geq 1}\frac{\ln Z_{X_n}(\alpha)}{n},
\ee
with  $\ln d\leq \Lambda_X(\alpha)\leq \alpha + \ln (2d)$. 
Moreover $\Lambda_X$ is convex,  non-decreasing and continuous on $[0,\infty)$, and the convergence is uniform on compact subsets of $[0,\infty)$.
\end{lem}
\begin{rems}\label{lbza}
i) One gets therefore
$
Z_{X_n}(\alpha)\geq e^{n\Lambda_X(\alpha)},
$
and 
\be\label{alphac}
\lim_{n\ra\infty}S_n(\alpha)=\infty \ \ \forall \alpha>\alpha_c\geq 0, \ \mbox{where } \  \Lambda_X(\alpha_c)=\ln (2d).
\ee
ii) The {\em connective constant} of the paths in the set $X$ is given by 
\be\label{connect}
\lim_{n\ra\infty} 
\Big(\sum_{\overrightarrow{x}\in X_n}1\Big)^{1/n}
= e^{\Lambda_{X}(0)}\in [d,2d].
\ee The connective constant on $\Z^d$ for the set $SP$, respectively $SAW$, will be denoted by $s(d)$, respectively $\mu(d)$.
\end{rems}

Let us introduce two auxiliary quantities, the 2-point function and the susceptibility, that appear in the analysis of partition functions  in the statistical mechanics of polymers.\\

Considering paths in $X_n$, with $X=SP$ or $X=SAW$ again, the 2-point function $G_\alpha(z,x)$  is defined in our setup  for $z\geq 0$ and $\alpha\geq 0$ by
\be\label{2pf}
G_\alpha(z,x)= \sum_{n\geq 0}z^n\sum_{{\overrightarrow{x}\in X_n}\atop {x_n=x}}e^{\alpha |x|}
\ee
and the the susceptibility $\chi_\alpha(z)$ is given by
\be\label{susc}
\chi_\alpha(z)=\sum_{n\geq 0}z^n\sum_{{\overrightarrow{x}\in X_n}}e^{\alpha |x_n|}=\sum_{x\in \G^d}G_\alpha(z,x).
\ee
Note that $|x_n|\leq n$ implies that the sum over $x\in \G^d$ is finite. The foregoing and the root test show that the radius of convergence of
\be
\chi_\alpha(z)=\sum_{n\geq 0}z^{n}Z_{X_n}(\alpha)
\ee
is given by the continuous, decreasing log-concave function of $\alpha$
\be\label{zcalpha}
z_c(\alpha)=\frac{1}{e^{\Lambda_X(\alpha)}}>0.
\ee
Given the interpretation of $\chi_\alpha(z)$, this characterises $z_c(\alpha)$ as a critical point.
Moreover, 
\be
\lim_{z\ra z^-_c(\alpha)}\chi_\alpha(z)=\infty,
\ee
Indeed, we get from Remark \ref{lbza} that for $0\leq z<z_c(\alpha)$
\be
\chi_\alpha(z)\geq \sum_{n\geq 0}z^{n}e^{n\Lambda_{X}(\alpha)} = \frac{1}{1-ze^{\Lambda_{X}(\alpha)}},
\ee
where the lower bound tends to infinity as $z\ra z^-_c(\alpha)$. Hence, we can relate the behaviour of $Z_{X_n}(\alpha)/(2d)^n$ for $n$ large to that of $\chi_\alpha(z/(2d))$ for $z$ close to $z_c(\alpha)$:
\begin{lem}\label{critacrit}
Let $\alpha_c$ be defined by $\Lambda_X(\alpha_c)=\ln(2d)$, see (\ref{alphac}). We have the characterisation
\be
\alpha < \alpha_c \ \Leftrightarrow \  z_c(\alpha) > 1/(2d) \ \Leftrightarrow \ \chi_\alpha(1/(2d))<\infty \ \Leftrightarrow \ \lim_{n\ra\infty}Z_{X_n}^{1/n}(\alpha)/2d <1.
\ee
\end{lem}
\subsubsection{SAW on $\cT_{2d}$}
The self avoiding walks on the tree are not numerous so that we can compute explicitly the partition function $Z_{SAW_n}(\alpha)$ and the quantities that derive from it. Any self avoiding path of length $n$ in $\cT_{2d}$ , $\overrightarrow{x}\in SAW_n$, is such that $|x_n|=n$, and there are $2d(2d-1)^{n-1}$ such paths. Hence
\be\label{sawtree}
Z_{SAW_n}(\alpha)=\frac{2d}{2d-1}((2d-1)e^\alpha)^{n}, \ \ \mbox{and } \ \ \Lambda_{SAW}(\alpha)=\ln (2d-1)+\alpha, \ \ \mbox{on} \ \cT_{2d}.
\ee
The function $\Lambda_{SAW}$ is affine, hence convex, and we get a value of $\alpha_c$ according to (\ref{alphac}), and a lower bound on the localisation length which read 
\be\label{sawbound}
\alpha_c=\ln\left(\frac{2d}{2d-1}\right)=\frac{1}{2d}+O(d^{-2})  \ \Rightarrow \  \cL\geq 2d+O(1), \ \ \mbox{when } \ d\ra \infty, \ \ \mbox{on} \ \cT_{2d}.
\ee
Since for any given $x\in\cT_{2d}$, there exists a unique SAW of length $n=|x|$ from $e$ to $x$, we have the following explicit expressions for $G_\alpha(z,x)$ and $\chi_\alpha(z)$
\begin{align}
&G_\alpha(z,x)= (ze^{\alpha})^{|x|}, \ \  \forall \ z\in \R  \\
&\chi_\alpha(z)= \frac{2d}{(2d-1)(1-z(2d-1)e^{\alpha})}, \ \ \forall \ 0<z<e^{-\alpha}/(2d-1),
\end{align}
which lead to the same conclusion about $\alpha_c$.\\
We improve the bound (\ref{sawbound}) in the next Section, discussing the susceptibility $\chi_\alpha(z)$ instead of the partition function.

\subsection{Lower Bound on $S_n(\alpha)$ on $\cT_{2d}$ via $Z_{SP_n}(\alpha)$ }\label{improvedtree}
We estimate $S_n(\alpha)$ by the sum $Z_{SP_n}(\alpha)=\sum_{\overrightarrow{x}\in SP_n}e^{\alpha |x_n|}$ on $\cT_{2d}$, taking into account a larger subset of paths in $SP_n$ than $SAW_n$, which allows to improve the lower bound on $\cL$.
\begin{thm}\label{locletree}
Let $U_\omega$ be a balanced quantum walk  on $\cT_{2d}$, with $d\geq 2$. Then\\
$
\limsup_{n\ra\infty}\E(\|e^{\alpha |X|/2}U_\omega^n \ e\otimes\tau_0\|^2)=\infty 
$
if 
\be
\alpha\geq \ln\left(\frac{1-1/(2d)+1/(2d)^2}{(1-1/(2d))(1-1/(2d)^2)}\right)=\frac{1}{2d^2}+O\left(\frac{1}{d^3}\right) \ \mbox{as} \ d\ra\infty, \nonumber
\ee
so that the localisation length $\cL$ satisfies 
\be
\cL > 2d^2 +O\left(d\right), \ \mbox{as} \ d\ra\infty.
\ee
\end{thm}
\begin{rem}
This result shows that dynamical localisation of a BRQW on a symmetric tree can only hold for larger and larger localisation length for higher and higher coordination number. Again, BRQWs are far from the RQWs known to displaying delocalisation.
\end{rem}
\begin{proof}
Let $\overrightarrow{x_k}$ be one of the $2d(2d-1)^{k-1}$ SAW of length $k\geq 1$. We shall attach at $1\leq j\leq k+1$ distinct sites of 
$\overrightarrow{x_k}$, decorations of variable lengths $1\leq r_i$, $i\in\{1, 2, \dots, j\}$ consisting of SAW of length $r_i$ that do not intersect $\overrightarrow{x_k}$. These decorations  will be run back and forth in sequence,  so that such paths belong to SP by construction. 
Given $r_i\geq 1$, there are $2(d-1)(2d-1)^{r_i-1}$ ways to construct these decorations, the first factor taking into account the fact they need to avoid the ingoing and outgoing directions of $\overrightarrow{x_k}$ at the site where they are attached. We neglect the extra possibilities at both ends of $\overrightarrow{x_k}$. There are $\begin{pmatrix} k+1 \cr j \end{pmatrix}$ ways to choose the $j$ sites of $\overrightarrow{x_k}$ and, given $\rho\geq 1$, there are 
$\begin{pmatrix} \rho-1 \cr j-1 \end{pmatrix}$ compositions of $\rho$  with $j$ parts, {\it i.e.} distinct ordered sets $\{r_1,r_2,\dots , r_j\}$ such that $\rho=r_1+r_2+\dots + r_j$, with $1\leq r_i\leq \rho$. The total length of such a decorated path is thus $n=k+2\rho$, so that their contribution to $ Z_{SP_n}(\alpha)$ reads 
\be
 Z_{SP_n}(\alpha) \geq \frac{2d}{2d-1}\sum_{{{\rho\geq 1, k\geq 1}\atop{1\leq j\leq k+1}}\atop{\mbox{\tiny s.t.}\ n=k+2\rho}}
(e^\alpha(2d-1))^k\left(\frac{2(d-1)}{(2d-1)}\right)^j(2d-1)^\rho
\begin{pmatrix} k+1 \cr j \end{pmatrix}\begin{pmatrix} \rho-1 \cr j-1 \end{pmatrix}.
\ee
Note that the summand is non negative, and that the constraints on $j$ and $\rho$ can be taken care of  by the  binomial coefficients. 
At this point it is convenient to consider the susceptibility $\chi_\alpha(z)$, see (\ref{susc}), which thus admits the lower bound  for $z\geq 0$
\be
\chi_\alpha(z)\geq c(d) \sum_{n\geq 0}  z^{n}\sum_{\rho\geq 1, k\geq 1, j\geq 1}  \delta_{n,k+2\rho}\
(e^\alpha(2d-1))^k\left(\frac{2(d-1)}{(2d-1)}\right)^j(2d-1)^\rho
\begin{pmatrix} k+1 \cr j \end{pmatrix}\begin{pmatrix} \rho-1 \cr j-1 \end{pmatrix},
\ee
where $c(d)$ is an inessential constant that can change from line to line. Performing the sums of $k$ and $\rho$ first, making use twice of the identity 
\be
\sum_{k\geq j}\begin{pmatrix} k \cr j \end{pmatrix} x^k=\frac{x^j}{(1-x)^{j+1}}, \ \mbox{if} \ |x|<1, 
\ee
we get
\be
\chi_\alpha(z)\geq c(d) \sum_{j\geq 1}
\left(\frac{ze^\alpha(2d-1)}{1-ze^\alpha(2d-1)}\right)^j\left(\frac{2(d-1)}{(2d-1)}\right)^j\left(\frac{z^2(2d-1)}{1-z^2(2d-1)}\right)^j,
\ee
provided 
\be\label{condza}
z^2(2d-1)<1, \ \mbox{and} \ ze^\alpha(2d-1)<1.
\ee 
Divergence of $\chi_\alpha(z)$ thus holds if
\begin{align}\label{divtree}
\frac{z^3e^\alpha (2d-1)2(d-1)}{(1-ze^\alpha(2d-1))(1-z^2(2d-1))}\geq 1 \ \Leftrightarrow e^\alpha\geq \frac{1-z^2(2d-1)}{(1-z^2)z(2d-1)}.
\end{align}
One checks that the second estimate above is compatible with the second condition (\ref{condza}) for all values of $z > 0$ and $d>1$, whereas the first condition holds in particular for $z=1/(2d)$ for any $d\geq 1$. Thus, in the limit of large $d$, (\ref{divtree}) implies divergence of $\chi_\alpha(1/(2d))$ for
\be
e^\alpha\geq 1+1/(2d^2)+O(1/d^3),
\ee
hence the announced upper bound on $\alpha_c$.
\ep
\end{proof}
\begin{rem}
Actually, considering paths with decorations of length 1 only is enough to improve the lower bound on $\cL$ (\ref{sawbound}) to a constant times $d^2$, whereas paths with decorations of fixed length strictly larger than one do not provide such an improvement. 
\end{rem}

\subsection{Lower Bound on $S_n(\alpha)$ on $\Z^{d}$}\label{lbszd}

We now turn to the cubic lattice case in dimension larger than or equal to two. We have the choice in the norm we use in (\ref{laplace}), and the critical value $\alpha_c$ beyond which $\limsup_{n\ra\infty}\E(\|e^{\alpha |X|/2}U_\omega^n \ e\otimes\tau_0\|^2)=\infty$ depends on this choice. 

If $|\cdot|_{a}\leq \nu |\cdot|_{b}$, and $\alpha_c(a), \alpha_c(b)$ denote the corresponding critical values, we have
\be
\alpha_c(b)\leq \nu \alpha_c(a).
\ee
This is a consequence of the fact that $\alpha \mapsto \E(\|e^{\alpha |X|/2}U_\omega^n \ e\otimes\tau_0\|^2)$ is increasing.
Thanks to the relations between $l^p$ norms on $\Z^d$, 
\be
|x|_p\leq |x|_q\leq d^{1/q-1/p}|x|_p, \ \ \mbox{for any $1\leq q<p\leq \infty$, and any $x\in\Z^d,$}
\ee  
a critical value $\alpha_c(\infty)$ obtained for $p=\infty$ will do for all $p\in [1,\infty[$, while an estimate for $\alpha_c(1)$ is easier to get. 

We investigate here the dependence on dimension of the quantities involved so far, in order to get a bound on $\alpha_c(1)$.
In the sequel, it will sometimes be useful to write $x=(L,y)\in \Z^d$, with $L\in\Z$ and $y\in \Z^{d-1}$, so that $|(L,y)|_1=|L|+|y|_1$.
Dropping the subscript SAW and emphasising the dependence on the dimension in the notation, the partition function (\ref{partfun}) writes 
\be
Z^{(d)}_n(\alpha)=\sum_{L\in\Z \atop y\in \Z^{d-1}}\sum_{{\overrightarrow{x}\in SAW_n}\atop x_n=(L,y)}e^{\alpha |(L,y)|_1}=
\sum_{L\in\Z \atop y\in \Z^{d-1}}\sum_{{\overrightarrow{x}\in SAW_n}\atop x_n=(L,y)}e^{\alpha (|L|+|y|_1)}.
\ee
For $L\neq 0$ fixed, we can further restrict summation to SAWs $\overrightarrow{x}$ in $\Z^d$ of length $n$ from the origin to $(L,y)$ that stem from a SAW $\overrightarrow{y}$ in $\Z^{d-1}$ of length $n-|L|$ from the origin to $y$ in the following way:\\
For $r=1, \dots, n-|L|$, one attaches at $r$ of the $n-|L|+1$ sites of $\overrightarrow{y}$ a straight segment along the first coordinate axis, of length $n_i\geq 1$, $i=1,\dots, r$, such that $\sum_{i=1}^r n_i=|L|$.  We observe that given $|L|\geq 1$, and $1\leq r\leq |L|$ there are $\begin{pmatrix}|L|-1 \\ r-1\end{pmatrix}$ compositions of $|L|$ with $r$ parts, {\it i.e.} ways to write $|L|=n_1+n_2+\cdots +n_r$, where $(n_1,\cdots, n_r)$ is an ordered set, with $n_j\geq 1$. 
Since the number of ways to choose the $r$ locations on $\overrightarrow{y}$ is $\binom{n-|L|+1}{r}$, Vandermonde's identity 
\be
\sum_{k=0}^n\begin{pmatrix}p \\ k\end{pmatrix}\begin{pmatrix}q \\ n-k\end{pmatrix}=\begin{pmatrix}p+q \\ n\end{pmatrix},
\ee
shows that, given a SAW $\overrightarrow{y}$ in $\Z^{d-1}$, there are $\binom{n}{|L|}$ SAWs in $\Z^d$ from the origin to $(L,y)$ constructed this way. If $L=0$, we just consider $\overrightarrow{y}$ in $\Z^{d-1}$ as a SAW in  $\Z^{d}$.
Hence, 
restricting attention to this subset of SAWs, we deduce
\be
\sum_{{\overrightarrow{x}\in SAW_n \atop x_n=(L,y)}}1\geq \binom{n}{|L|}\sum_{{\overrightarrow{y}\in SAW_{n-|L|} \atop y_n=y}}1
\ee
so that
\be\label{esone}
Z^{(d)}_n(\alpha)\geq \sum_{{L\in\Z\atop y\in \Z^{d-1}}}e^{\alpha  |L|}e^{\alpha  |y|_1}\binom{n}{|L|}\sum_{{\overrightarrow{y}\in SAW_{n-|L|} \atop y_n=y}}1\equiv 
\sum_{{L\in\Z}}e^{\alpha  |L|}\binom{n}{|L|}Z^{(d-1)}_{n-|L|}(\alpha).
\ee 
This last estimate yields 
\begin{thm}\label{thmone}
For the $l^1$ norm on $\Z^d$, we have the estimate for any $d\in\N$,
\be
\alpha_c(1)\leq \ln(2).
\ee
\end{thm}
\begin{proof}
From $Z^{(d)}_n(\alpha)\geq e^{n\Lambda^{(d)}(\alpha)}$, the right side of (\ref{esone}) is bounded below by $(e^\alpha+e^{\Lambda^{(d-1)}(\alpha)})^n$, so that, see (\ref{free}),
\be
\Lambda^{(d)}(\alpha)\geq \ln (e^\alpha+e^{\Lambda^{(d-1)}(\alpha)}).
\ee
Using the starting point $\Lambda^{(1)}(\alpha)=\alpha$, see (\ref{sawtree}), induction yields $\Lambda^{(d)}(\alpha)\geq \ln (d e^\alpha).$
With the criterion (\ref{alphac}), we get $\alpha_c(1)\leq \ln(2)$.
\end{proof}
\begin{rem}
For the $l^p$ norm with $\infty>p\geq 1$, using $|(L,y)|_p\geq \frac{|L|+|y|_p}{2^{1-1/p}}$, the same argument yields an upper bound on $\alpha_c(p)$ 
which, however, is not as good as $\ln(2)$.
\end{rem}
For the sake of comparison, it is worth mentioning that for the $l^\infty$ norm, the same strategy also provides an explicit bound. Indeed, using  $|(L,y)|_\infty=\max(|L|,|y|_\infty)$, we get
 \be
Z^{(d)}_n(\alpha)\geq \sum_{{y\in \Z^{d-1}\atop L\in\Z}\atop |L|\leq |y|_\infty}e^{\alpha  |y|_\infty}\binom{n}{|L|}\sum_{{\overrightarrow{y}\in SAW_{n-|L|} \atop y_n=y}}1+\sum_{{y\in \Z^{d-1}\atop L\in\Z}\atop  |L|> |y|_\infty}e^{|L|\alpha}\binom{n}{|L|}\sum_{{\overrightarrow{y}\in SAW_{n-|L|} \atop y_n=y}}1.
\ee
Now, using $e^{\alpha |y|_\infty}\geq e^{\alpha|L|}$ in the first sum yields 
\be
Z^{(d)}_n(\alpha)\geq \sum_{{y\in \Z^{d-1}\atop L\in\Z}}e^{\alpha  |L|}\binom{n}{|L|}\sum_{{\overrightarrow{y}\in SAW_{n-|L|} \atop y_n=y}}1\equiv 
\sum_{{y\in \Z^{d-1}\atop L\in\Z}}e^{\alpha  |L|}\binom{n}{|L|}Z^{(d-1)}_{n-|L|}(0).
\ee
Then, with  $Z^{(d-1)}_{n-|L|}(0)\geq \mu(d-1)^{n-|L|}$, see (\ref{connect}), we obtain
\be
Z^{(d)}_n(\alpha)\geq \left(e^\alpha + \mu(d-1)\right)^n \ \Rightarrow \ \Lambda^{(d)}(\alpha)\geq \ln (e^\alpha + \mu(d-1)).
\ee
Hence, together with the known asymptotics $\mu(d)=2d-1+O(1/d)$ in high dimension, see {\it e.g.} \cite{MS, BDGS}, we obtain
\be
\alpha_c(\infty)\leq \ln(2d-\mu(d-1))=\ln(3)+O(1/d).
\ee
Not only is this bound not as good as the one on $\alpha_c(1)$ for $d$ large, but, since $d\leq \mu(d) \leq 2d$, the method used here to estimate $\alpha_c(\infty)$ cannot yield a better estimate than $\ln(2)$, the bound for $\alpha_c(1)$.

\section{Localisation length and correlation length of SAW}\label{seccor}

In this closing section, we make the link between the quantities involved in the analysis of the localisation length of BRQWs and the correlation length appearing in the analysis of SAWs in $\Z^d$. Our concern being in the large $d$ limit, we assume $d\geq 5$ in this section. We recall the relevant notions and results to be found in \cite{MS, BDGS}, for example. 
The two point function for SAWs, for paths starting at the origin, is given by (\ref{2pf}) at $\alpha=0$, and similarly for the corresponding susceptibility
\begin{align}
G_0(z,x)=\sum_{n\geq 0}z^n\sum_{{\overrightarrow{x}\in SAW_n}\atop {x_n=x}}1,\ \ \mbox{and} \ \
\chi_0(z)=\sum_{x\in\Z^d}G_0(z,x).
\end{align}
If $x\neq 0$, both functions have radius of convergence $z_c(0)=1/\mu(d)$, for SAWs, see (\ref{zcalpha}).\\

The inverse correlation length or mass for SAWs in $\Z^d$ is defined for all $z>0$ by
\be
m_d(z)=\liminf_{L\ra \infty} \frac{-\ln G_0(z,(L,0))}{L}, 
\ee
where we stress the dependence on the dimension $d$ in the notation. 
We provide below alternative expressions of the inverse correlation length $m_d(z)=1/\xi_d(z)$.
The mass $m_d$ enjoys several properties proven in \cite{MS}, that we summarise in the next proposition:
\begin{prop}
For any $0<z<1/\mu(d)$,  $z\mapsto m_d(z)$ is real analytic, strictly decreasing, concave in $\ln z$, and 
\be
\lim_{z\downarrow  0}m_d(z)=\infty, \ \ \lim_{z\uparrow 1/\mu(d)}m_d(z)=0.
\ee 
For $z>1/\mu(d)$, $m_d(z)=-\infty$ and, if $d\geq 5$, $m_d(1/\mu(d))=0$.
\end{prop}
In order to make the link with the localisation length $\cL$, we introduce for all $z>0$,
\begin{align}
G_L(z)=\sum_{y\in\Z^{d-1}}G_0(z, (L,y)),
\end{align}
the generating functions of SAWs from the origin to the plane $\{x_1=L\}$. 
This generating function is known to satisfy, among other things,
\begin{lem}\label{glz} For all $0<z<1/\mu(d)$, 
\begin{align}
&m_d(z)=\lim_{L\ra \infty} \frac{-\ln G_L(z)}{L}=\sup_{L\geq 1}\frac{-\ln G_L(z)}{L}, \ \ \mbox{and} \\
&e^{-m_d(z)L}\leq G_L(z)\leq \chi_0(z)^2e^{-m_d(z)L}, \ \ \mbox{if } \ L\geq 1.
\end{align}
\end{lem}
This is enough to prove the sought for relationship between localisation length and correlation length:
\begin{prop}\label{corloc} On the cubic lattice $\Z^d$, the localisation length of BRQWs is bounded below by the correlation length of SAWs  at the critical value of simple random walks:
\be
\cL\geq \xi_d(1/(2d))  \ \Leftrightarrow \ \alpha_c\leq m_d(1/(2d)).
\ee
\end{prop}
\begin{rems}
i) If one can show that $\lim_{d\ra\infty} m_d(1/(2d))=0$, it would prove the divergence with the dimension of the localisation length $\cL$. 
Despite the fact that $1/(2d)$ is the critical point for simple random walks, to which $1/\mu(d)$ converges, the control of the limit seems to be non trivial according to experts in SAWs, and it is not clear that the limit vanishes.\\
ii) An uncontrolled prediction on $\lim_{d\ra \infty }\xi_d(1/(2d))$ is however not encouraging: for $d\geq 5$ fixed, it is known that $\xi_d(z)\simeq \{\frac{D}{2d}(\frac{z_c}{z_c-z})\}^{1/2}$, as $z\ra z_c=1/\mu(d)$, where $D$ is the diffusion constant of SAWs, such that $1+c_1/d\leq D \leq 1+ c_2/d$, for large $d$, with $c_1, c_2>0$, see Thm 6.1.5 and Proposition 6.2.11 in \cite{MS}. Ignoring the fact that error terms are {\em a priori} not uniform in $d$, and plugging in the asymptotics 
$\mu(d)=2d-1+O(1/d)$, we get $\xi_d(1/(2d))\simeq 1$.  However, Theorem \ref{thmone} yields the better bound $\cL\geq 1/\ln(2)$.
\end{rems}
\begin{proof}
Considering $\chi_\alpha(z)=\sum_{n\geq 0} z^{n}\sum_{\overrightarrow{x}\in SAW_n}e^{\alpha |x_n|_\infty}$ and paths with end points of the form $x_n=(L,y)$, we get with $|(L,y)|_\infty\geq |L|$,  
\be
\chi_\alpha\left(\frac{z}{2d}\right)=\sum_{n\geq 0}\left(\frac{z}{2d}\right)^{n}\sum_{L\in\Z \atop y\in \Z^{d-1}}e^{\alpha |(L,y)|_\infty}\sum_{\overrightarrow{x}\in SAW_n \atop x_n=(L,y)}1\geq \sum_{L\in\Z }e^{\alpha |L|}\sum_{y\in \Z^{d-1}}G_0\left(\frac{z}{2d},(L,y)\right).
\ee
The last sum above coincides with $G_L\left(\frac{z}{2d}\right)$, so that restricting to $L\geq1$, Lemma \ref{glz} yields the estimate
\be
\chi_\alpha\left(\frac{z}{2d}\right)\geq \sum_{L>0}e^{\alpha L}e^{-m_d\left(\frac{z}{2d}\right)L}.
\ee
The right hand side being convergent if and only if $\alpha<m_d\left(\frac{z}{2d}\right)$, we deduce the result from Lemma \ref{critacrit}.\ep
\end{proof}

\section{Appendix}

{\bf Proof of Lemma \ref{stand}:} \\
 For the log-convexity property we compute the second derivative of $\Lambda_{X_n}(\alpha)= \frac1n\ln Z_{X_n}(\alpha)$:
\be
\Lambda_{X_n}'(\alpha)= \frac{Z'_{X_n}(\alpha)}{nZ_{X_n}(\alpha)}, \ \ \ \ \Lambda_{X_n}''(\alpha)=\frac{Z_{X_n}(\alpha)Z''_{X_n}(\alpha)-(Z'_{X_n}(\alpha))^2}{nZ^2_{X_n}(\alpha)},
\ee
with $Z^{(p)}_{X_n}(\alpha)=\sum_{\overrightarrow{x}\in SP_n} |x_n|^pe^{\alpha |x_n|}$, for $p\in\N$. Hence $\Lambda_{X_n}'(\alpha)>0$ and 
\bea
Z_{X_n}(\alpha)Z''_{X_n}(\alpha)-(Z'_{X_n}(\alpha))^2&=&\sum_{ {\overrightarrow{x}\in X_n} \atop {\overrightarrow{y}\in X_n}} (|x_n|^2-|x_n||y_n|)e^{\alpha (|x_n|+|y_n|)}\nonumber\\
&=&\frac12\sum_{ {\overrightarrow{x}\in X_n} \atop {\overrightarrow{y}\in X_n}} (|x_n|-|y_n|)^2e^{\alpha(|x_n|+|y_n|)}\geq 0.
\eea
To prove the existence of the limit as $n\ra\infty$ of $\Lambda_{X_n}(\alpha)$, we use a subadditivity argument. 
For  $\overrightarrow{x}\in X_n$ and $\overrightarrow{y}\in X_m$, the concatenated path $\overrightarrow{xy}$ obtained by following $\overrightarrow{x}$ from $e$ to $x_n$ and then $\overrightarrow{y}$ from $x_n$ to $x_ny_m$ does not necessarily satisfy the requirements to belong to $X_{n+m}$. By the triangular inequality, $|x_ny_m|\leq |x_n|+|y_m|$, we have for $\alpha\geq 0$,
\bea
Z_{X_{n+m}}(\alpha)&=&\sum_{\overrightarrow{x}\in X_{n+m}}e^{\alpha |x_{n+m}|}\leq \sum_{ {\overrightarrow{x}\in X_n} \atop {\overrightarrow{y}\in X_m}}e^{\alpha |x_ny_{m}|}\nonumber\\
&=& \sum_{ {\overrightarrow{x}\in X_n} \atop {\overrightarrow{y}\in X_m}}e^{\alpha |x_n|}e^{\alpha |y_{m}|}=Z_{X_{n}}(\alpha)Z_{X_{m}}(\alpha).
\eea
Consequently, $\{\Lambda_{X_n}(\alpha)=\frac1{n}\ln(Z_{X_{n}}(\alpha))\}_{n\in\N^*}$ is a subadditive sequence, which implies 
(\ref{free}) for each $\alpha\geq 0$. The upper and lower bounds on $\Lambda_X(\alpha)$ are consequences of $S_n(\alpha)\leq e^{n\alpha}$ and (\ref{lb}).
Since $\Lambda_{X_n}$ is convex and increasing for all $n\in\N^*$, we immediately get that $\Lambda_{X}$ is convex and non decreasing on $[0,\infty)$. Moreover, $\Lambda_{X_n}$ converges uniformly to $\Lambda_{X}$ on any compact set of $]0,\infty[$, and $\Lambda_{X}$ is continuous on $]0,\infty[$, see \cite{S}. It remains to extend the result to compact sets of the form $[0,b]$, $b>0$. Let $0<\alpha<1$ and write $\alpha=0(1-\alpha)+1\alpha$. By convexity and monotony of $\Lambda_{X}$, we have
\be
0\leq \Lambda_{X}(\alpha)-\Lambda_{X}(0)\leq \alpha (\Lambda_{X}(1)-\Lambda_{X}(0)),
\ee
so that $\Lambda_{X}$ is continuous at $0$. Since a sequence of monotonous functions that converges to a continuous function on a compact set converges uniformly, this finishes the proof.
\ep

 {

}
\end{document}